\documentclass[11pt]{article}
\usepackage[utf8]{inputenc}

\usepackage{amsmath}
\usepackage{amssymb}

\usepackage[margin=3cm]{geometry}

\usepackage{graphicx} 

\usepackage{array}

\newtheorem{theorem}{Theorem}
\newtheorem{corollary}[theorem]{Corollary}

\newtheorem{lemma}[theorem]{Lemma}
\newtheorem{claim}[theorem]{Claim}
\newtheorem{definition}[theorem]{Definition}

\newtheorem{remark}[theorem]{Remark}

\newenvironment{proof}{\noindent\bf{Proof.}\rm}{\hfill$\blacksquare$\bigskip}

\newcommand{\items}{\mathcal{M}} 
\newcommand{\agents}{\mathcal{N}} 

\usepackage[]{color-edits}
\addauthor{UF}{red}

\addauthor{YG}{blue}

\begin{document}

\title{On best of both worlds allocations with subadditive valuations}

\author{Uriel Feige\thanks{Weizmann Institute, Israel. {\tt uriel.feige@weizmann.ac.il}}}

\date{27 September, 2026}

\maketitle

\begin{abstract}
We consider allocation of indivisible goods to agents with equal entitlements and subadditive valuations. As an ex-post fairness notion we consider the maximin share (MMS), and as an ex-ante fairness notion we consider the maximum expectation share (MES), which is always at least as large as the MMS, and sometimes much larger. We present a simple transformation that for every $0 < \rho  \le 1$, given any algorithm that produces $\rho$-MMS allocations, transforms it into a randomized allocation algorithm that offers $\rho$-MMS ex-post simultaneously with $\eta$-MES ex-ante. We prove  several new properties of MES, and use them to show that $\eta \ge \min[\frac{\rho}{2 + \rho}, \frac{1}{4}]$.  We also present cases in which the transformation results in a higher value of $\eta$. 

Applying our transformation to currently known allocation algorithms shows for subadditive valuations the existence of randomized allocations that are simultaneously $\Omega(\frac{1}{\log\log n})$-MES ex-ante and $\Omega(\frac{1}{\log\log n})$-MMS ex-post, and for XOS valuations the existence of randomized allocations that are simultaneously $\frac{4}{27}$-MES ex-ante and $\frac{4}{17}$-MMS ex-post. 
\end{abstract}




\section{Introduction}

There are many settings in which one seeks fair division of resources. These include dividing an inheritance, settling a divorce, distributing humanitarian aid, scheduling medical procedures, registering students to classes of limited capacity, allocating housing units to eligible residents, and more. 
We consider a standard model of fair allocation of indivisible items to agents with equal entitlements.  There is a set $\items$ of $m$ items, and there is a set $\agents$ of $n$ agents. An allocation $A = (A_1, \ldots, A_n)$ is a partition of the set $\items$ into $n$ disjoint subsets (referred to as {\em bundles}), such that for each $i \in \agents$, agent $i$ receives bundle $A_i$.  A randomized allocation is a distribution over allocations.

Each agent $i\in {\cal{N}}$ has a valuation function (also referred to just as valuation, for short) $v_i: 2^{\items} \rightarrow R$ that assigns a value to each possible subset of items, with $v_i(\emptyset) = 0$. Valuations are assumed to be monotone non-decreasing ($v_i(S) \le v_i(T)$ if $S \subset T$), signifying that we view items as {\em goods} (items that agents want to receive). 
We will mostly be interested in the following classes of valuations in the {\em complement free} hierarchy of~\cite{LLN}, stated in increasing generality:

\begin{itemize}
    \item Additive valuations: $v(S) = \sum_{e \in s} v(e)$.
    \item Submodular valuations: $v(S \cup \{e\}) - v(S) \ge v(T \cup \{e\}) - v(T)$ for every $S \subset T$ and $e \in \items$.
    \item XOS valuations: maximum over a collection of additive valuations.
    \item Subadditive valuations: $v(S) + v(T) \ge v(S \cup T)$.
\end{itemize}

Additive valuations is the class most commonly studied. However, when $m$ is large, classes higher up the hierarchy appear to better model realistic valuations. For example, when distributing humanitarian aid such as food packages to individuals in need, it is unlikely that a hungry individual will view 1000 tomatoes as 1000 times more valuable than a single tomato. 

Fairness constraints are constraints that allocations (or distributions over allocations) need to satisfy in order to be classified as fair. We consider fairness constraints towards each agent separately, giving $n$ sets of fairness constraints. An allocation is {\em acceptable for agent $i$} if it satisfies the fairness constraints of agent $i$, and it is {\em acceptable} if it satisfies the fairness constraints of all agents.


The fairness constraints that we consider are {\em share-based}, meaning that the fairness constraints of agent $i$ are based only on the set $\items$ of items, her own valuation function $v_i$, and the set $A_i$ of items allocated to $i$. The constraint is expressed via a share function $s(\items, v_i, n)$ that specifies a minimum value that $i$ expects to receive. 

We distinguish between ex-post and ex-ante fairness. 

Ex-post constraints refer to a single allocation, the one output by an allocation algorithm. The allocation $A = A_1, \ldots, A_n$ satisfies the share-based fairness constraint if $v_i(A_i) \ge s(\items, v_i, n)$. 

There are cases in which using randomized allocation algorithms is more desirable than selecting a single allocation. For example, this is the case when there are fewer items than agents. A deterministic allocation surely gives some agent no item, whereas randomized allocation algorithms can offer every agent a positive probability of getting at least one item.
Ex-ante constraints refer to a distribution over allocations that is generated by a randomized allocation algorithm. In this paper (as in most related work), these constraints refer to the {\em expected value} that an agent receives, if an allocation is drawn at random from the distribution generated by the randomized allocation algorithm. A distribution $D$ over allocations satisfies the share-based fairness constraint if $E_{(A_1, \ldots, A_n) \sim D}[v_i(A_i)] \ge s(\items, v_i, n)$.

In share-based {\em best of both worlds} (BoBW) type of results, we have two types of shares: an ex-ante share $s_a$, and an ex-post share $s_p$. We seek a distribution $D$ over allocations such that $D$ guarantees to each agent an expected value (ex-ante) of at least her $s_a$ share value, and $D$ is supported only over allocations that guarantee to each agent at least her $s_p$ share value (ex-post). This notion makes sense only if the ex-ante guarantee is stronger than the ex-post guarantee (because ex-post guarantees hold also ex-ante). Consequently, we aim to have ex-ante shares $s_a$ whose value is at least as high as the ex-post share $s_p$, and potentially much larger. 

We now introduce the two shares considered in this paper, the {\em maximin share} (MMS) ex-post (attributed to~\cite{Budish11}), and the {\em maximum expectation share} (MES) ex-ante (a related concept, without using the name MES, was introduced in~\cite{FeigeT14} in the context of fair division of a divisible good).
Conceptually, the MMS and MES are based on the same principle, and the only difference is that the MMS applies it to ex-post guarantees whereas MES applies it to ex-ante guarantees. The principle is to consider what can be guaranteed to an agent if all agents are indistinguishable (have the same entitlements and same valuation functions). The MMS is the best possible guarantee for the realized value ex-post. The MES is the best possible guarantee for the expected value ex-ante.  An immediate consequence is that the MES is at least as high as the MMS, and this holds for every possible valuation function (including also non-monotone valuations).

\begin{definition}
    \label{def:MMS}
    The {\em maximin share} (MMS) is the maximum value that an agent can secure for herself if she were to partition the set of items into $n$ parts, and receive the least valuable part.
$$MMS(\items, v_i, n) = \max_{A = (A_1, \ldots, A_n)} \min_{j \in \{1, \ldots, n\}} v_i(A_j),$$ 
where $A$ ranges over all partitions of $\items$ into $n$ parts. 
\end{definition}

\begin{definition}
\label{def:MES}
    The {\em maximum expectation share} (MES) is the maximum expected value that an agent can secure for herself if she were to partition the set of items into $n$ parts, and receive one of the parts uniformly at random.
$$MES(\items, v_i, n) = \max_{A = (A_1, \ldots, A_n)} E_{j \in \{1, \ldots, n\}} [v_i(A_j)],$$ 
where $A$ ranges over all partitions of $\items$ into $n$ parts, and expectation is taken over the uniform distribution over the parts.
\end{definition}

As the partition that maximizes the MMS is also a candidate partition for the MES, and MES takes the average value of bundles in a partition whereas MMS takes the minimum value, we always have that:

$$MES(\items, v_i, n) \ge MMS(\items, v_i, n)$$

The gap between the MES and the MMS can be unbounded, as the MMS can be~0 (this is always the case if $m < n$), whereas the MES is always positive (unless of course $v_i$ is identically~0).

For comparison, we mention here the {\em proportional share} (PS), defined as
$$PS(\items, v_i, n) = \frac{1}{n} \cdot v_i(\items).$$ 

The MES is always at least as large as PS, as we may take $A_1 = \items$ and keep the rest of the parts empty. Thus,

$$MES(\items, v_i, n) \ge PS(\items, v_i, n)$$

For additive valuations, we have that:

$$MES(\items, v_i, n) = PS(\items, v_i, n)  \ge MMS(\items, v_i, n)$$

However, for submodular valuations and beyond, the value of the proportional share might be  much smaller than the MES, and even much smaller than the MMS. (Consider for example $n$ identical items that are substitutes of each other: every non-empty set of items has value~1. The proportional share is $\frac{1}{n}$, but the MMS is~1.) 


We say that a share is {\em feasible} for a class of valuations if every allocation instance within the class has an allocation acceptable under this share (each agent receives at least her share value). The MMS is not feasible, not even for additive valuations~\cite{KPW18}. The PS is feasible ex-ante (give each agent the whole set of items with probability $\frac{1}{n}$) for all classes of valuations. The MES is feasible ex-ante for additive valuations (as it equals the PS in this case), but not feasible for submodular valuations and beyond.

Due to non-feasibility of MMS and MES, we (and others) consider approximations of shares. 

\begin{definition}
    For a share SH and a value $\rho \le 1$, a $\rho$-SH allocation is one that gives each agent at least a $\rho$ fraction of her SH value.
\end{definition}

Our goal in this paper is to design randomized allocation algorithms that provide good approximations to the MES ex-ante, and good approximation to the MMS ex-post. Results that require simultaneous ex-ante and ex-post guarantees are sometimes referred to as {\em best of both worlds} results (even if the results are not necessarily best possible), a term introduced in~\cite{FSV20BoBW}.

\subsection{Known results}

Table~\ref{tab:MMS} summarizes the values of $\rho$ for which $\rho$-MMS is currently known to be feasible ex-post for all $n$ (tighter lower bounds are known for small values of $n$). 

\begin{table}[hbt!]
        \centering
\begin{tabular}{ | m{4 cm} || m{4 cm}| m{4 cm} | } 
 \hline
 \multicolumn{3}{|c|}{Feasible $\rho$ for $\rho$-MMS ex-post} \\
 \hline
  & lower bound & upper bound  \\ 
 \hline
  \hline
 Additive \newline valuations & $\frac{7}{9}$ \; \; \cite{huangZhou2025} & $\frac{39}{40}$  \; \; \cite{FST21Negative}  \\
 \hline
Submodular \newline valuations & $\frac{10}{27}$ \; \; \cite{BUF23Submodular} & $\frac{2}{3}$   \; \; \cite{CCMS25Subadditive4}  \\ 
 \hline
 XOS  \newline valuations & $\frac{4}{17}$ \; \; \cite{FG25APSSubadditive} & $\frac{1}{2}$ \; \; \cite{GhodsiHSSY22BeyondAdditive}    \\ 
 \hline
 Subadditive \newline valuations & $\Omega(\frac{1}{\log\log n})$ \; \; \cite{feigeMulti} & $\frac{1}{2}$  \; \; \cite{GhodsiHSSY22BeyondAdditive}   \\ 
 \hline
\end{tabular}
        \caption{
        $n$ denotes the number of agents in the allocation instance.  Some previous bounds appear in~\cite{KPW18, amanatidis2017MMS34, BarmanK20Submodular, GhodsiHSSY22BeyondAdditive, BarmanK20Submodular, AkramiMSS23XOSBoBW, SS24SubadditiveXOS, FeigeHuang25, SS25Subadditive}, and more. Improved results for small values of $n$ appear in~\cite{amanatidis2017MMS34, FeigeNorkin22, CCMS25Subadditive4}, and more.}
        \label{tab:MMS}
\end{table}

For ex-ante guarantees, tight bounds are known for the classes of valuations of interest in this paper, except for submodular valuations. These bounds appear in Table~\ref{tab:MES}. 

\begin{table}[hbt!]
        \centering
\begin{tabular}{ | m{4 cm} || m{4 cm} | } 
 \hline
 \multicolumn{2}{|c|}{Feasible $\eta$ for $\eta$-MES ex-ante} \\
  \hline
 Additive \newline valuations & $1$   \\
 \hline
Submodular \newline valuations & $1 - \frac{1}{e} + \epsilon$, for some $\epsilon > 0$   \\ 
 \hline
 XOS  \newline valuations & $1 - \frac{1}{e}$    \\ 
 \hline
 Subadditive \newline valuations & $\frac{1}{2}$   \\ 
 \hline
\end{tabular}
        \caption{
         The bounds are taken from~\cite{FG25APSSubadditive}, where it is explained how they follow from earlier work. For XOS and subadditive valuations, this earlier work is~\cite{Feige09Subadditive}. For submodular, this earlier work is~\cite{FeigeV10}, which implies that $\epsilon > 0$ and gives examples where $\epsilon < 0.15$.}
        \label{tab:MES}
\end{table}

There have been a few best of both worlds type results.
For additive valuations, in the equal entitlements case, it is known that there are randomized allocations that simultaneously are ex-ante proportional and ex-post $\frac{1}{2}$-MMS. The proof is based on an extension of a certain {\em faithful randomized rounding} technique applied to a linear programming representation of the randomized proportional allocation~\cite{BEF22BoBW}. 
For XOS valuations and equal entitlements, it is known how to achieve $\frac{1}{4}$-MMS ex-ante simultaneously with $\frac{1}{8}$-MMS ex-post~\cite{AkramiMSS23XOSBoBW}. However, we would like ex-ante guarantees with respect to the MES, not the MMS.

Our focus in this work is on share-based guarantees. There is also work that considers simultaneous ex-post and ex-ante comparison-based guarantees (envy freeness and relaxations of it)~\cite{FSV20BoBW, Aziz20BoBW, AGM23BoBW, HSV23BoBW, FeldmanMNP24}. For additive valuations, strong comparison based guarantees also imply strong share based guarantees. However, for sub-additive valuations this is no longer true. Even just with submodular valuations, there are allocation instances in which envy free allocations offer only $\frac{1}{n}$-MMS.

\subsection{Our results}
\label{sec:results}

We present a straightforward ``black-box" transformation that transforms any allocation algorithm ALG that gives only ex-post guarantees to a randomized allocation algorithm RALG that offers also ex-ante guarantees. Our main technical contribution is the analysis of the ex-ante guarantees of this transformation. Among other things, this involves (see Section~\ref{sec:preliminary}) formulating and proving new properties of the MES.  

Suppose that for some $\rho$, ALG offers each agent at least $\rho \cdot MMS(\items, v_i, n)$ ex-post. We wish RALG to preserve the ex-post guarantee, and to also offer each agent at least $\eta \cdot MES(\items, v_i, n)$ ex-ante, for $\eta$ that is as large as possible as a function of $\rho$. Given a choice of $\eta$, we say that an item $e$ is {\em large} for an agent $i$ if $$v_i(e) \ge \max[\rho \cdot MMS(\items, v_i, n), \; \eta \cdot MES(\items, v_i, n)],$$ and {\em small} otherwise. Our randomized allocation algorithm RALG uses $\rho$ and $\eta$ as input parameters. It has two phases.

\begin{enumerate}
    \item Take a uniformly random permutation over $\agents$. 
(It suffices to fix one arbitrary permutation, and pick a permutation at random with uniform probability over the cyclic shifts of the fixed permutation.)
Then, each agent $i$
in her turn considers the item $e$ of highest value (according to $v_i$, breaking ties arbitrarily) among those remaining.  If $e$ is large for $i$, then $i$ takes $e$ and ends her participation in RALG. If $e$ is small for $i$, then $i$ does not take the item, and stays for the second phase of RALG.

\item Let $\items'$ and $\agents'$ denote the sets of items and agents that remain after the first phase. Run ALG with $\items'$ and $\agents'$. 
\end{enumerate}

It is quite easy to see that regardless of the value of $\eta$, RALG as described above preserves the ex-post guarantee of $\rho$-MMS of ALG. The question is what is the largest value of $\eta$ (as a function of $\rho$) for which if RALG is run with this input parameter $\eta$, it indeed guarantees each agent at least $\eta$-MES.  

\begin{theorem}
\label{thm:blackbox}
    For any $0 \le \rho \le 1$, let $\eta = \min[\frac{\rho}{2 + \rho}, \frac{1}{4}]$. 
    Let ALG be an arbitrary allocation algorithm that offers every agent $i$ with subadditive valuation $v_i$ at least $\rho \cdot MMS(\items, v_i, n)$ (ex-post). Then RALG as described above preserves the ex-post guarantee, and also guarantees at least $\eta \cdot MES(\items, v_i, n)$ ex-ante. The same holds for every subclass of subadditive, such as XOS.
\end{theorem}

\begin{remark}
\label{rem:smalln}
     At the end of Section~\ref{sec:blackbox} we present an example showing that $\frac{\rho}{2 + \rho}$ is essentially the best bound that can be shown for our transformation, if one treats ALG as a blackbox. However, there is room for minor improvements. Specifically, the bound in Theorem~\ref{thm:blackbox} is somewhat stronger than stated, giving $\eta = \min[\frac{\rho n}{2n - 1 + (n - 1)\rho}, \frac{1}{4}]$. 
\end{remark}

Interestingly, Theorem~\ref{thm:blackbox} is not a tradeoff result between ex-post and ex-ante guarantees. Rather, it shows that some ex-ante guarantees can be given without any loss in the ex-post guarantees. The ex-ante value of $\eta$ promised in Theorem~\ref{thm:blackbox} does not seem to be very good, but one should keep in mind that without the transformation implied by the theorem, the value of $\eta$ might be much worse, even~0. 

Theorem~\ref{thm:blackbox} and Remark~\ref{rem:smalln} have the following immediate corollaries.

\begin{corollary}
\label{cor:blackbox}
    For XOS valuations, there are randomized allocations that are ex-post $\frac{4}{17}$-MMS and ex-ante $\frac{2}{19}$-MES. For subadditive valuations, there are randomized allocations that are ex-post $\Omega(\frac{1}{\log\log n})$-MMS and ex-ante $\Omega(\frac{1}{\log\log n})$-MES. If $n = 4$, there are randomized allocations that are ex-post $\frac{1}{2}$-MMS and ex-ante $\frac{4}{17}$-MES (and $\frac{1}{4}$-MES if $n \in \{2,3\}$). 
\end{corollary}

\begin{proof}
    For XOS valuations, apply the transformation of Theorem~\ref{thm:blackbox} on the allocation algorithm of~\cite{FG25APSSubadditive}. For subadditive valuations, apply it on the allocation algorithm of~\cite{feigeMulti}. When $n \le 4$, apply the transformation of Remark~\ref{rem:smalln} on the allocation algorithm of~\cite{CCMS25Subadditive4}.
\end{proof}

In some cases, RALG works with higher values of $\eta$ than those stated in Theorem~\ref{thm:blackbox}. We illustrate this for XOS valuations, for which we obtain an ex-ante guarantee that is stronger than the one appearing in Corollary~\ref{cor:blackbox}.

\begin{theorem}
\label{thm:whitebox}
    For XOS valuations, there are randomized allocations that are $\frac{4}{17}$-MMS ex-post and $\frac{4}{27}$-MES ex-ante.
\end{theorem}


\section{Proofs of our main results}

In our proof, we start with an allocation algorithm ALG that offers ex-post guarantees, and transform it into a randomized allocation algorithm RALG that preserves the ex-post guarantees, and adds to them ex-ante guarantees. Though our transformation is quite simple, the analysis of the ex-ante guarantees turns out to be more delicate than one would initially expect. Some inequalities that would be convenient to have do not hold for the MES. For example, consider a setting in which a single item $e$ is given to some agent $j$, and agent $j$ leaves. How does this affect agent $i$? For the MMS, regardless of the value $v_i(e)$, the MMS for agent $i$ does not decrease. For the MES, if $v_i$ is additive (in which case MES is the proportional share), then if $v_i(e) > MES_i$ the MES decreases, and if $v_i(e) \le MES_i$ the MES does not decrease (and even increases, if the inequality is strict). But if $v_i$ is subadditive (the case we address in our work), the MES might decrease even if $v_i(e) < MES_i$. See Section~\ref{sec:preliminary} for an example.


Thus, before embarking on the proofs of our main theorems, we need to develop an understanding of properties of the MES.



\subsection{Properties of the MES}
\label{sec:preliminary}

We define the {\em welfare} function $W$ of a valuation $v$ as the maximum welfare (sum of values) that can be given to $n$ agents that all have valuation $v$.

$$W(\items, v_i, n) = \max_{A = A_1, \ldots, A_n} \sum_j v_i(A_j)$$

Using the welfare function, the definition of MES resembles that of the proportional share.

$$MES(\items, v_i, n) = \frac{1}{n}W(\items, v_i, n)$$

We now compare properties of the MMS and the MES. Like the MMS, the MES is non-increasing in $n$.

\begin{claim}
\label{cl:monotone}
    For every monotone valuation $V$ and $n' \le n$,
$$MES(\items, v, n') \ge MES(\items, v, n)$$
\end{claim}

\begin{proof}
    Keep only the $n'$ parts of highest value from the $n$-partition giving $MES(\items, v, n)$ its value.
\end{proof}

Let us recall a well known property of the MMS.

\begin{claim}
    \label{cl:MMSgiveItem}
     For every monotone valuation $v$, and $k < n$ and every set $S \subset \items$ of size $k$ it holds that:
     $$MMS(\items \setminus S, v, n-k) \ge MMS(\items, v, n)$$
\end{claim}

\begin{proof}
     As only $k$ items are removed by $S$, at least $n-k$  bundles of the  $MMS(\items, v, \frac{1}{n})$ partition  are still in $\items \setminus S$. They imply the inequality $MMS(\items \setminus S, v, n-k) \ge MMS(\items, v, n)$.
\end{proof}

A claim similar to Claim~\ref{cl:MMSgiveItem} fails for the MES. Moreover, already when $v$ is submodular, it fails even if $e$ has value much smaller than the MES. Here is a simple example illustrating this. Suppose that $n=2$, there are three items, $e_1$ and $e_2$ that are substitutes of each other $v(\{e_1, e_2\}) = v(e_1) = v(e_2) = 1$, and $e_3$ that has value $x$, and also marginal value $x$ with respect to any set to which it is added. Then $(\{e_1, e_3\}, \{e_2\})$ is an MES partition showing that $MES(\items, v, 2) = 1 + \frac{x}{2}$. Removing $e_3$ we have that $MES(\items \setminus \{e_3\}, v, 1) =  v(\{e_1, e_2\}) = 1$. 

Due for examples such as the one above, we need to settle for a claim that appears to be quite weak.

\begin{claim}
\label{cl:giveItem}
    For every sub-additive valuation $v$, and $k < n$ and every set $S \subset \items$ of size $k$ it holds that:

    $$MES(\items \setminus S, v, n-k) \ge MES(\items, v, n) - \frac{1}{n}\sum_{e\in S} v(e)$$
\end{claim}

\begin{proof}
    Let $A_1, \ldots, A_n$ be an optimal $MES(\items, v, n)$ partition, meaning that $MES(\items, v, n) = \frac{1}{n}\sum_{j=1}^n v(A_j)$. For each $j$, let $A'_j$ be $A'_j = A_j \setminus S$. Using subadditivity of $v$ and the fact that each item from $S$ appears only in one set $A_j$, we have that $\sum_{j=1}^n v(A'_j) \ge \sum_{j=1}^n v(A_j) - \sum_{e\in s} v(e)$. Hence, the average value (over all choices of $j$) of a set $A'_j$ is at least $ MES(\items, v, n) - \frac{1}{n}\sum_{e\in s} v(e)$. By considering the $n - k$ highest valued sets from this $n$-partition, we conclude that $MES(\items \setminus S, v, n-k) \ge MES(\items, v, n) - \frac{1}{n}\sum_{e\in s} v(e)$. (By monotonicity of $v$, leftover items can be added in an arbitrary manner to the $n-k$ parts.)
\end{proof}

We now turn to compare between the values of the MES and MMS. 
By definition, the MES is always at least as large as the MMS. The MMS might not provide any approximation to the MES. If there are fewer items than agents, the MMS is~0, whereas the MES is always at least $\frac{1}{n} \cdot v_i(\items)$. 
We observe that for sub-additive valuations (and subclasses of sub-additive), either the MES is of the same order of magnitude as the MMS, or there are single items whose value is at least a constant fraction of the MES. (This is not true for valuations beyond sub-additive.)

\begin{lemma}
    \label{lem:MMSversusMES}
    For any subadditive valuation $v$, every $k$ and every $\beta > 0$, if  $v(e) \le \beta \cdot MES(\items, v, n)$ for every item $e \in \items$, then  $$MMS(\items, v, k) \ge \alpha_{k, \beta} \cdot MES(\items, v, n),$$ where
   $$\alpha_{k,\beta} = \frac{n - (k-1)\beta}{n + k -1} > \frac{n - k\beta}{n+k}$$
   In particular, $$MMS(\items, v, n) > \frac{1 - \beta}{2} \cdot MES(\items, v, n).$$
\end{lemma}

\begin{proof}
    To simplify notation, without loss of generality we may assume that $MES(\items, v, n) = 1$.
    
    Let $A_1, \ldots, A_n$ be a partition of $\items$ that satisfies $\sum_{j=1}^n v(A_i) = n$. Using this partition, we shall create a new partition $B_1, \ldots , B_k$ with $v(B_j) \ge \alpha_{k, \beta}$ for every $j$. This will serve as an MMS partition that proves the lemma.

    For each part $A_j$  of the MES partition do the following ``bundle filling" procedure. Open a fresh bundle and place items one by one in the bundle, until the first time in which the value of the bundle reaches or exceeds $\alpha$. At that point we declare this bundle ``full", add it as a bundle of the MMS partition, and restart the process with the items that still remain in $A_j$ and a new fresh bundle to fill. The process ends when no items remain in $A_j$. At that point, there might be a ``partial" bundle whose value has not yet reached $\alpha$. 

    When the above process ends for all $A_j$, we have some full bundles, and at most $n$ partial bundles (at most one partial bundle for each part of the MES partition). We need to show that the number of full bundles is at least $k$. Equivalently, we need to show that $\sum_{j = 1}^t n_j \ge k$, where $n_j$ is the number of full bundles produced from $A_j$.

Every partial bundle has value smaller than $\alpha$. Every full bundle has value at most $\alpha + \beta$. This follows from subadditivity of $v$, because prior to receiving its last item it had value at most $\alpha$, and the last item had value at most $\beta$.

    Consider any bundle $A_j$. By subadditivity, its value is at most $n_j \cdot (\alpha + \beta) + \alpha$. Thus, $n \le \sum_{j \in [n]} v(A_i) \le (\sum_{j = 1}^t n_j) \cdot (\alpha + \beta) + n \cdot \alpha$.  In any negative example (meaning that $\alpha$ is too large to produce $k$ full bundles) $\sum_{j = 1}^t n_j < k$, and then we have that $n < (n + k - 1)\alpha + {(k-1)\beta}$. Thus, if 
    $$\alpha \le \frac{n - (k-1)\beta}{n+k-1}$$
    there cannot be a negative example, and then $MMS(\items, v, k) \ge \alpha MES(\items, v, n)$.
\end{proof}

For XOS (and sub-additive) valuations, the value of $\alpha_{k,\beta}$ in Lemma~\ref{lem:MMSversusMES} cannot be improved upon if $c \cdot \beta = \alpha_{k,\beta}$ for an integer $c$. Suppose that $MES(\items, v, n)$ is just slightly below~1. The negative examples are ones in which $n-1$ bundles of the $MES(\items, v, n)$ partition have value slightly below $\alpha_{k, \beta}$, and the remaining bundle has $kc - 1$ items of value $\beta$. With $\alpha_{k,\beta} = \frac{n - (k-1)\beta}{n + k -1}$, $MES(\items, v, n)$ is just slightly below~1, whereas $MMS(\items, v, k)$ (and also the $APS(\items, v, k)$, see definition in Section~\ref{sec:XOS}) have value  $\alpha_{k, \beta}$.

A proof technique somewhat similar to that of Lemma~\ref{lem:MMSversusMES} was used under the name of {\em set splitting} in~\cite{DobzinskiLRV24}, though the context there is different, and the analysis there would not give sufficiently strong bounds in our context.

\subsection{A blackbox transformation}
\label{sec:blackbox}

Here we prove Theorem~\ref{thm:blackbox}

\begin{proof}
Let ALG be an allocation algorithm for instances with subadditive valuations, and let $\rho > 0$ be such that ALG outputs $\rho$-MMS allocations. Recall from Section~\ref{sec:results} the randomized allocation algorithm RALG and the notion of large items. We need to determine how to set the input parameter $\eta$ for RALG.

We first note that the ex-post guarantees of RALG do not depend on $\eta$.

\begin{claim}
\label{cl:expost2}
    Every agent $i$ 
    gets ex-post a value of at least $\rho \cdot MMS(\items, v_i, n)$.
\end{claim}

\begin{proof}
By the definition of large items, it suffices to consider the case that $i$ did not get a large item. In this case, $i \in \agents'$. 
  By Claim~\ref{cl:MMSgiveItem}, $MMS(\items', v_i, n') \ge MMS(\items, v_i, n)$. By the guarantees of ALG, agent $i$ gets value at least $\rho \cdot MMS(\items', v_i, n') \ge \rho \cdot MMS(\items, v_i, n)$.  
\end{proof}

We now analyse the ex-ante guarantees. These clearly depend on the choice of $\eta$. Let $\agents_1$ denote those agents $i$ for which $$\rho \cdot MMS(\items, v_i, n) < \eta\cdot MES(\items, v_i, n).$$

\begin{claim}
\label{cl:phase2}
 Every agent $i \not\in \agents_1$ gets at least $\eta$-MES ex-ante.   
\end{claim}

\begin{proof}
    The ex-post guarantee of $\rho \cdot MMS(\items, v_i, n)$ from Claim~\ref{cl:expost2} implies also an ex-post guarantee of $\eta \cdot MES(\items, v_i, n)$, because  $i \not\in \agents_1$. Every ex-post guarantee holds also ex-ante.
\end{proof}

We now analyse the ex-ante guarantees of an agent $i \in \agents_1$. Note that as $i \in \agents_1$, an item $e$ is large for $i$ if $v_i(e) \ge \eta \cdot MES(\items, v_i, n)$. 
Let $L_i$ denote the set of items that are large for $i$, and let $H_i$ denote the set of items that are ``huge" for $i$, of value at least $2\eta \cdot MES(\items, v_i, n)$ (every huge item is also large). Define $V(H_i)$ as $V(H_i) = \sum_{e \in H_i} v_i(e)$. 
We define $\gamma$ to be $\gamma = 1 - \frac{V(H_i)}{W(\items, v_i, n)}$. The value of $\gamma$ is a lower bound on the fraction of welfare (recall definition of welfare in Section~\ref{sec:preliminary}) that is left in $\items$ if all huge items are removed.

\begin{claim}
    \label{cl:exante}
    Suppose that $\eta \le \frac{1}{4}$. Consider agent $i\in \agents_1$, and assume that her valuation is scaled so that $MES(\items, v_i, n) = 1$. Let $p_h$ denote the probability that in the first phase $i$ receives an item from $H_i$. Then  $i$ gets ex-ante at least the minimum of the following options.
    
    \begin{enumerate}
        \item $\eta$.
        \item $\max[1 - \gamma, 2p_h \eta] + (1 - p_h)\min[\eta, \frac{\rho \cdot (\gamma - \eta)}{2}]$.
    \end{enumerate}
\end{claim}

\begin{proof}
Let $A_1, \ldots, A_n$ denote the MES partition for agent $i$. By the assumption that $MES(\items, v_i, n) = 1$, we have that the total welfare $W_i$ satisfies $\sum_j v_i(A_j) = n$.


We consider separately the contribution of $H_i$ towards the ex-ante value of $i$, and the contribution of $\items \setminus H_i$.

For the contribution of $H_i$ we have two lower bounds.

If $i$ gets an item from $H_i$, then she gets value at least $2\eta$. As $p_h$ denotes the probability of getting such an item, this gives a lower bound of $2p_h \eta$. This lower bound implies that we can assume that $p_h \le \frac{1}{2}$, as otherwise the minimum of the two option in the claim is $\eta$ from option~1.

The assumption that $p_h < 1$ implies that $|H_i| < n$, because otherwise $i$ would surely get an item from $H_i$ in the first phase. Then, for every $k \le |H_i|$, agent $i$ has probability $\frac{1}{n}$ of being the $k$th agent in the random permutation of the first phase, and then get one of her top $k$ items in $H_i$. Thus, in expectation, she gets value at least $\frac{1}{n} \sum_{e \in H_i} v_i(e)$. By the definition of $\gamma$, this is $(1 - \gamma) \cdot MES$. Note that this lower bound allows us to assume that $\gamma \ge 1 - \eta$. 

So far, we established a lower bound of $\max[1 - \gamma, 2p_h \eta]$ on the contribution of the items in $H_i$. We now turn to consider the contribution of the items of $M \setminus H_i$. There are two possibilities for the source of this contribution, and the minimum of the two applies. One possibility is that $i$ gets an item in the first phase. This item will necessary be from $L_i \setminus H_i$, and will have value at least $\eta$. The other possibility is that $i$ reaches the second phase of RALG. We now analyse the contribution of the second phase. Importantly,  Claim~\ref{cl:MMSgiveItem} does not hold for the MES, and hence cannot be used in our analysis. 

Let $W_i^{-H} = \sum_j v_i(A_j \setminus H_i)$ denote the welfare that remains after all items of $H_i$ are removed. By definition of $\gamma$, $W_i^{-H} \ge \gamma n$. Note that the number of remaining agents decreased to $n - |H_i|$.
In addition, there were other items taken in the first phase (in particular, all items of $L_i \setminus H_i$). Let $t$ denote their number. As each of these items has value at most $2\eta$, the welfare $W'_i$ that remains is at least $W'_i \ge \gamma n - 2t\eta$, and the number of agents that remain is $n' = n - |H_i| - t$.  

Consider the set of items that remain at the beginning of the second phase, and denote it by $\items'$. 
No item $e \in \items'$ has value above $\eta$. Scale $v_i$ to $v'_i = \frac{n}{W'_i} \cdot v_i$, so that $W(\items', v'_i, b) \ge n$. Relative to $v'_i$, no item has value larger than $\eta \cdot \frac{n}{W'_i}$. Applying Lemma~\ref{lem:MMSversusMES} (we use here the relaxed equality with that has the term $k$ rather than $k-1$), we then have that 
$$MMS(\items', v'_i , n') \ge \frac{n - n' \frac{n}{W'_i}\eta}{n + n'}$$
which after replacing $v'_i$ by $v_i$, and using $n' \le n - t$ and $W'_i \ge \gamma n - 2t\eta$, gives
$$MMS(\items', v_i , n') \ge \frac{W'_i - n' \eta}{n + n'} \ge \frac{(\gamma - \eta)n - t\eta}{2n - t}$$
As Claim~\ref{cl:exante} assumes that $\eta \le \frac{1}{4}$ and the proof so far showed that we can assume that $\gamma \ge 1- \eta$, the above expression is minimized when $t=0$. Hence $MMS(\items', v_i , n') \ge \frac{(\gamma - \eta)}{2}$.
    
As RALG runs ALG on $(\items', \agents')$, agent $i$ gets at least $\rho \cdot MMS(\items', v_i, n') \ge \rho \frac{(\gamma - \eta)}{2}$. This completes the proof of Claim~\ref{cl:exante}.
\end{proof}

Considering the two options in Claim~\ref{cl:exante}, we show that if the second option is smaller, namely $\max[1 - \gamma, 2p_h \eta] + (1 - p_h)\min[\eta, \frac{\rho \cdot (\gamma - \eta)}{2}] \le \eta$, then the second option is lower bounded by $\frac{\rho(1 - \eta)}{2}$. Ignoring the term $1 - \gamma$, the second option is lower bounded by a convex combination of the term $2\eta$ and the term $\min[\eta, \frac{\rho \cdot (\gamma - \eta)}{2}]$. For this convex combination to be smaller than $\eta$, the minimum for the second term is attained in $\frac{\rho \cdot (\gamma - \eta)}{2}$. Also, we can lower bound $\max[1 - \gamma, 2p_h \eta]$ by the average $\frac{1 - \gamma}{2} + p_h \eta$. Thus the second option simplifies to 

$$\frac{1 - \gamma}{2} + p_h \eta + (1 - p_h)\frac{\rho \cdot (\gamma - \eta)}{2}$$

As we can assume that $\frac{\rho \cdot (\gamma - \eta)}{2} < \eta$, the expression is minimized when $p_h = 0$. As $\rho, \gamma \le 1$, the expression is minimized when $\gamma = 1$. Hence, the second option in Claim~\ref{cl:exante} is lower bounded by $\frac{\rho(1 - \eta)}{2}$.
For the choice of $\eta = \frac{\rho}{2 + \rho}$ (as in the statement of Theorem~\ref{thm:blackbox}), options 2 is lower bounded by $\eta$, as desired. Hence, agents in $\agents_1$ get the required ex-ante guarantee. This, in combination with Claims~\ref{cl:expost2} and~\ref{cl:phase2}, completes the proof of Theorem~\ref{thm:blackbox}.
\end{proof}

In order to prove Remark~\ref{rem:smalln}, in the proof of Claim~\ref{cl:exante} we use the tighter bound (with $k-1$, not $k$) of Lemma~\ref{lem:MMSversusMES}, and get for the second phase the guarantee of $\frac{\rho(n - (n-1)\eta}{2n-1}$. This gives $\eta = \frac{\rho n}{2n-1 + \rho(n-1)}$.

Finally, we provide an example showing that the blackbox approach of Theorem~\ref{thm:blackbox} sometimes gives bounds no better than $\eta = \frac{\rho}{2 + \rho}$ (up to terms that tend to~0 as $n$ increases). The example is the one following the proof of Lemma~\ref{lem:MMSversusMES}, with $k = n$. Recall that in that example $\beta$ needs to divide $\alpha_{n,\beta}$. We assume that $\rho$ is such that it will lead to this divisibility requirement. We set $\beta$ slightly smaller than $\frac{\rho}{2 + \rho}$ which gives roughly $\alpha_{n, \beta} = \frac{1}{2 + \rho}$. Then it might be that no agent would like to take an item in the first phase, and the agents enters the second phase with an MMS value that $\frac{1}{2 + \rho}$-MES. If we treat ALG as a blackbox, then all we can assume is that the agent gets at least a $\rho$ fraction of her MMS, which is roughly $\frac{\rho}{2 + \rho}$-MES. In this respect, the choice of $\eta = \frac{\rho}{2 + \rho}$ in Theorem~\ref{thm:blackbox} is best possible (up to terms that tend to~0 as $n$ increases).

\subsection{Looking inside the blackbox}
\label{sec:XOS}

We recall the allocation algorithm of~\cite{FG25APSSubadditive} that for XOS valuations gives $\frac{4}{17}$-MMS allocations. In fact, it gives $\frac{4}{17}$-APS allocations. APS, the anyprice share~\cite{BEF24APS}, is a share designed for instances with arbitrary entitlements, in which each agent has entitlement $b_i > 0$, and $\sum b_i = 1$.
The APS can be viewed as a fractional analog of the MMS.  

\begin{definition}
    \label{def:APS}
    A fractional partition $FP$ of $\items$ is a collection $(S_1, S_2, \ldots)$ of subsets of items with associated positive weights $(\lambda_1, \lambda_2, \ldots)$ with $\sum_j \lambda_j = 1$. $FP$ is $b_i$-balanced if for every item $e\in \items$ it holds that $\sum_{j | e \in S_j} \lambda_j = b_i$. Then:
$$APS(\items, v_i, b_i) = \max_{FP} \min_{S \in FP} v_i(S),$$
where $FP$ ranges over all $b_i$-balanced fractional partitions of $\items$.
\end{definition}

 For entitlements of the form $\frac{1}{n}$ (which is required for the MMS to be defined), $APS(\items, v_i, \frac{1}{n}) \ge MMS(\items, v_i, n)$ always holds.

We use the following lemma, proved in~\cite{GhodsiHSSY22BeyondAdditive} for MMS, and extends to APS, as observed in~\cite{FG25APSSubadditive}. 

\begin{lemma}
    \label{lem:XOS}
    For agents with XOS valuations, there is an allocation that gives every agent $i$ value at least $\frac{1}{2} \cdot MMS(\items, v_i, 2n)$, and moreover, at least $\frac{1}{2} \cdot APS(\items, v_i, \frac{1}{2n})$.
\end{lemma}

We now describe a framework for allocation algorithms used in~\cite{GhodsiHSSY22BeyondAdditive, AkramiMSS23XOSBoBW, FG25APSSubadditive}. To simplify notation, we assume that valuations are scaled so that $MMS(\items, v_i, n) = 1$ for every agent.  We call this algorithm $ALG_t$, as it is parameterized by a positive integer {$t$}. It gives an approximation ratio of $\rho_t$.

\begin{enumerate}
    \item Among the remaining items and agents, let $S$ be the smallest set of items (breaking ties arbitrarily) for which there is an agent $i$ with $v_i(S) \ge \rho_t$. If $|S| \le t$, then allocate $S$ to $i$, remove $S$ and $i$ from the allocation instance, and repeat. If $|S| > t$, move to the next step.
    \item Let $\items'$ and $\agents'$ denote the sets of items and agents that remain, and let $n' = |\agents'|$. Using the algorithm of Lemma~\ref{lem:XOS}, give each agent $i \in \agents'$ a bundle of value at least $\frac{1}{2} \cdot APS(\items', v_i, \frac{1}{2n'})$.
    \item Allocate the remaining items arbitrarily.
\end{enumerate}

To show that $ALG_t$ gives at least $\rho_t \cdot MMS$, one needs to prove that if $|\agents'| > 0$, then $APS(\items', v_i, \frac{1}{2n'}) \ge 2\rho_t \cdot MMS(\items, v_i, n)$ for every $i \in \agents'$. For $t \le 4$, this is known to hold for $\rho_t = \frac{t}{4t + 1}$~\cite{FG25APSSubadditive}. For $t = 4$, one gets $\frac{4}{17}$-APS allocations (and hence also $\frac{4}{17}$-MMS). 

In Theorem~\ref{thm:whitebox} we get stronger ex-ante bounds than those of Corollary~\ref{cor:blackbox} that use Theorem~\ref{thm:blackbox} as a blackbox. This is because $ALG_t$ ultimately makes use of $APS(\items', v_i, \frac{1}{2n'})$ instead of $MMS(\items', v_i, n')$. 
This allows us to improve the analysis as follows. For the purpose of analysing the ex-ante bound, we shall lower bound $APS(\items', v_i, \frac{1}{2n'})$ by $MMS(\items', v_i, 2n')$. The techniques in the current paper allow us to go directly from $MES(\items', v_i, n)$ to $MMS(\items', v_i, 2n')$. This gives better bounds that using a two step path, from $MES(\items', v_i, n)$ to $MMS(\items', v_i, n')$, and only then to $APS(\items', v_i, \frac{1}{2n'})$. The other improvement is that in the first step of $ALG_t$, agents pick up to $t$ items and not just one item (which is the case for Theorem~\ref{thm:blackbox}). 
This will allow us to use a version of Lemma~\ref{lem:MMSversusMES} with $\beta = \frac{\eta}{t}$, instead of with $\beta = \eta$. 

The proof of Theorem~\ref{thm:whitebox} (and Lemma~\ref{lem:MMSversusMESXOS} that replaces Lemma~\ref{lem:MMSversusMES}) appears in Section~\ref{app:XOS} in the appendix.

 \section{Discussion}

Theorem~\ref{thm:blackbox} has the appealing feature that any improvement in ex-post ratios $\rho$ will automatically lead also to an improvement in the ex-ante ratios $\eta$. Theorem~\ref{thm:whitebox} shows that the value obtained for $\eta$ is sometimes better than the value $\frac{\rho}{2 + \rho}$ guaranteed by Theorem~\ref{thm:blackbox}. However, a limitation of our approach is that it cannot give a value of $\eta$ that is larger than $\rho$. This is because if there are no large items, RALG simply runs ALG (and the MES is at least as large as the MMS). 
Note that we may assume that there are input allocation instances in which no item has value larger than $\rho$, and the approximation ratio of ALG on these instances is no better than $\rho$. Otherwise, also in the presence of large items we could get an approximation ratio better than $\rho$: we could give the large items to agents that desire them (this does not decrease the MMS of other agents), and then run ALG on the remaining instance.

There is only one case that we are aware of in which the transformation of ALG to RALG indeed gives $\eta = \rho$. This is the case for additive valuations (for which MES is the proportional share, PS), and for the ex-post fairness notion of TPS (the truncated proportional share~\cite{BEF22BoBW}) rather than the MMS. It is known that $\frac{n}{2n-1}$-TPS allocations always exist, and that the ratio of $\frac{n}{2n-1}$ is best possible. Our transformation gives randomized allocations that are at least $\frac{n}{2n-1}$-TPS ex-post and at least $\frac{n}{2n-1}$-PS ex-ante. We omit the proof. This tradeoff is technically incomparable, but arguably, less desirable, than that offered by a randomized allocation algorithm of~\cite{BEF22BoBW}, which offers at least $\frac{1}{2}$-TPS ex-post and at least the full proportional share ex-ante.

For submodular valuations, our transformation should lead to ex-ante ratios that are better than those implied by Theorem~\ref{thm:blackbox}, but we have not attempted to determine what these ratios are.

\subsection*{Acknowledgments}

This research was supported in part by the Israel Science Foundation (grant No. 1122/22). AI tools were used so as to detect typos and other minor slips in the presentation, but were not used at all in the process of obtaining the results in this paper.

\bibliographystyle{alpha}



\newcommand{\etalchar}[1]{$^{#1}$}

\begin{appendix}

\section{Stronger ex-ante guaranties for XOS valuations}
\label{app:XOS}

In this section we prove Theorem~\ref{thm:whitebox}. Before that, we state and sketch the proof of a variation on Lemma~\ref{lem:MMSversusMES}, that applies when valuations are XOS.

\begin{lemma}
    \label{lem:MMSversusMESXOS}
    For any XOS valuation $v$, every $k$ and every sufficiently small (to be defined below) $\eta > 0$, if no set of up to~$t$ items in a bundle of the given MES partition has value larger than $\eta \cdot MES(\items, v, n)$, then setting $\beta = \frac{\eta}{t}$ we have:  $$MMS(\items, v, k) \ge \alpha_{k, \beta} \cdot MES(\items, v, n),$$ where
   $$\alpha_{k,\beta} = \frac{n - (k-1)\beta}{n + k -1} > \frac{n - k\beta}{n+k}$$
   In this lemma, $\eta$ is considered sufficiently small if the inequality $\eta \le \alpha_{k,\beta}$ holds with the respective $\alpha_{k,\beta}$.
\end{lemma}

\begin{proof}
The proof is similar to that of Lemma~\ref{lem:MMSversusMES}, with the following changes.

In the proof of Lemma~\ref{lem:MMSversusMES}, for each part $A_j$  of the MES partition we do a ``bundle filling" procedure. Here, we order the items of $A_j$ prior to starting the bundle filling. As $v$ is XOS, there is an additive valuation $v'$ such that $v'(A_j) = v(A_j)$ and $v'(S) \le v(S)$ for every $S \subset A_j$. Fixing such a $v'$, we order the items on $A_j$ in non-increasing order of their $v'$ values. We use this order in the bundle filling procedure, and declare a bundle full at the first time in which the $v'$ value of the bundle reaches or exceeds our target value $\alpha_{k,\beta}$. Importantly, the sum of values of the top $t$ items in this order is at most $\eta \le \alpha_{k,\beta}$. Hence, for a bundle to overflow it must have at least $t+1$ items. As every item except for the first $t-1$ items has $v'$ value at most $\frac{\eta}{t}$, the overflow in any bundle is by at most $\frac{\eta}{t}$. Hence, here $\frac{\eta}{t}$ serves the same role that $\beta$ serves in the bounds of Lemma~\ref{lem:MMSversusMES}.
\end{proof}

We now prove Theorem~\ref{thm:whitebox}.

\begin{proof}
    Consider an allocation instance with XOS valuations. We transform $ALG_t$ into a randomized allocation algorithm $RALG_t$ 
using the transformation appearing before Theorem~\ref{thm:blackbox}. (Recall that in the description of $ALG_t$, we assume that valuations are scaled so that $MMS(\items, v_i, n) = 1$. In our transformation, when we get to run $ALG_t$, we retain the scaling $MMS(\items, v_i, n) = 1$, implying that for the instance on which $ALG_t$ is run we have $MMS(\items', v_i, |\agents'|) \ge 1$.)

We analyse the value of $\eta$ that $RALG_t$ can guarantee. Let $\agents_1$ denote those agents $i$ for which 
\begin{equation}
    \label{eq:N_1}
    \rho \cdot MMS(\items, v_i, n) < \eta\cdot MES(\items, v_i, n)
\end{equation}
Note that an item $e$ is large for $i \in \agents_1$ if $v_i(e) \ge \eta \cdot MES(\items, v_i, n)$. 

The main technical difference between the proof of Theorem~\ref{thm:whitebox} and the proof of Theorem~\ref{thm:blackbox} is in the analysis of the ex-ante guarantees of an agent $i \in \agents_1$. Hence, we focus on this part. First, we recall some notation also used in the proof of Theorem~\ref{thm:blackbox}.

Let $L_i$ denote the set of items that are large for $i$, and let $H_i$ denote the set of items that are ``huge" for $i$, of value at least $2\eta \cdot MES(\items, v_i, n)$. Define $V(H_i)$ as $V(H_i) = \sum_{e \in H_i} v_i(e)$. 
We define $\gamma$ to be $\gamma = 1 - \frac{V(H_i)}{W(\items, v_i, n)}$. 

The following is a variation of Claim~\ref{cl:exante}.

\begin{claim}
    \label{cl:exanteXOS}
    Suppose that $\eta \le \frac{4}{27}$. Consider agent $i\in \agents_1$, and assume that her valuation is scaled so that $MES(\items, v_i, n) = 1$. Let $p_h$ denote the probability that in the first phase $i$ receives an item from $H_i$. Then  $i$ gets ex-ante at least the minimum of the following options.
    
    \begin{enumerate}
        \item $\eta$.
        \item $\max[1 - \gamma, 2p_h \eta] + (1 - p_h)\min[\eta, \frac{2\gamma - \eta}{12}]$.
    \end{enumerate}
\end{claim}

\begin{proof}
Let $A_1, \ldots, A_n$ denote the MES partition for agent $i$. By the assumption that $MES(\items, v_i, n) = 1$, we have that the total welfare $W_i$ satisfies $\sum_j v_i(A_j) = n$.


We consider separately the contribution of $H_i$ towards the ex-ante value of $i$, and the contribution of $\items \setminus H_i$.

For the contribution of $H_i$, the proof is similar to that of Claim~\ref{cl:exante}, and we get a lower bound of $\max[1 - \gamma, 2p_h \eta]$. 

We now turn to consider the contribution of the items of $\items \setminus H_i$. There are two possibilities for the source of this contribution, and the minimum of the two applies. One possibility is that $i$ gets an item from $L_i \setminus H_i$ in the first phase, implying that she gets value at least $\eta$. The other possibility is that $i$ reaches the second phase of $RALG_t$. We now analyse the contribution of the second phase.

Let $W_i^{-H} = \sum_j v_i(A_j \setminus H_i)$ denote the welfare that remains after all items of $H_i$ are removed. By definition of $\gamma$, $W_i^{-H} \ge \gamma n$. Note that the number of remaining agents decreased to $n - |H_i|$.
In addition, there were other items taken in the first phase (in particular, all items of $L_i \setminus H_i$). Also, there were agents who each took up to $t=4$ items in the first step of $ALG_t$. Let $\ell$ denote the number of additional agents beyond $|H_i|$ that took items either in the first phase or in step~1 of $ALG_t$. Each such agent either took one item of value at most $2\eta$, or at most $t$ items, each of value at most $\eta$. The welfare $W'_i$ that remains is at least $W'_i \ge \gamma n - \ell  t\eta$, and the number of agents that remain is $n' = n - |H_i| - \ell$.

Consider the set of items that remain at the beginning of step~2 of $ALG_t$, and denote it by $\items'$. 
No item $e \in \items'$ has value above $\eta$. Moreover, for every bundle $A_j$ of the original MES partition, no set of up to~4 items in $A_j \cap \items'$ has value above $\eta$. (If such a set existed, then by inequality~(\ref{eq:N_1}) its value would exceed the threshold used in Step~1 of $ALG_t$, and the agent would not reach Step~2 of $ALG_t$.) 
Scale $v_i$ to $v'_i = \frac{n}{W'_i} \cdot v_i$, so that $W(\items', v'_i, b) \ge n$. Relative to $v'_i$, no set of up to~4 items in a bundle of the given MES partition has value larger than $\eta \cdot \frac{n}{W'_i}$. 
Applying 
Lemma~\ref{lem:MMSversusMESXOS} with $k = 2n'$ (the lemma applies because for our parameters, $\eta$ satisfies the "sufficiently small" condition) we get:
$$MMS(\items', v'_i , 2n') \ge \frac{n - 2n' \frac{n}{4W'_i}\eta}{n + 2n'}$$
which after replacing $v'_i$ by $v_i$, and using $n' \le n - \ell$ and $W'_i \ge \gamma n - 4\ell \eta$, gives
$$MMS(\items', v_i , 2n') \ge \frac{W'_i - \frac{n' \eta}{2}}{n + 2n'} \ge \frac{(\gamma - \frac{\eta}{2})n - \frac{7}{2}\ell \eta}{3n - 2\ell}$$
As Claim~\ref{cl:exanteXOS} assumes that $\eta \le \frac{4}{27}$, and we may also assume that $\gamma \ge 1- \eta$ (recall the proof of Lemma~\ref{lem:MMSversusMES}), the above expression is minimized when $\ell=0$. Hence $MMS(\items', v_i , 2n') \ge \frac{(2\gamma - \eta)}{6}$.
    
As $RALG_t$ runs the algorithm of Lemma~\ref{lem:XOS} on $(\items', \agents')$, agent $i$ gets at least $\frac{1}{2} \cdot MMS(\items', v_i, 2n') \ge \frac{(2\gamma - \eta)}{12}$. This completes the proof of Claim~\ref{cl:exanteXOS}.
\end{proof}
    
Considering the two options in Claim~\ref{cl:exanteXOS}, we show that if the second option is smaller, namely $\max[1 - \gamma, 2p_h \eta] + (1 - p_h)\min[\eta, \frac{2\gamma - \eta}{12}] \le \eta$, then the second option is lower bounded by $\frac{2 - \eta}{12}$. As in the proof of Theorem~\ref{thm:blackbox}, 
the second option simplifies to 

$$\frac{1 - \gamma}{2} + p_h \eta + (1 - p_h)\frac{2\gamma - \eta}{12}$$

As we can assume that $\frac{2\gamma - \eta}{12} < \eta$, the expression is minimized when $p_h = 0$. As $\gamma \le 1$, the expression is minimized when $\gamma = 1$. Hence, the second option in Claim~\ref{cl:exanteXOS} is lower bounded by $\frac{2 - \eta}{12}$.
For the choice of 
$\eta = \frac{4}{27}$ (as in the statement of Theorem~\ref{thm:whitebox}), option 2 is higher than 
$\eta$, as desired. Hence, agents in $\agents_1$ get the required ex-ante guarantee. 

Other parts of the proof of Theorem~\ref{thm:whitebox} are similar to the proof of Theorem~\ref{thm:blackbox}.
\end{proof}

\section{A note on unequal entitlements}

Theorem~\ref{thm:blackbox} considers settings in which agents have equal entitlements. One would like to have a similar theorem for settings with unequal entitlements. In such settings, we suggest using the anyprice share (APS) as the ex-post share (as already done in Section~\ref{sec:XOS}), and the version of MES defined in~\cite{FG25APSSubadditive} as the ex-ante share. The following argument illustrates a difficulty in handling settings with unequal entitlements. If there is an item of very high value, we must give each agent some chance of winning it (otherwise, there cannot be a good ex-ante guarantee). In the equal entitlements case, regardless of which agent receives the item, the MMS of the remaining agents does not decrease. However, in the unequal entitlements case, if the agent who wins the item has very low entitlement, the APS of the remaining agents might decrease, even to~0. This makes it challenging (though not necessarily impossible) to give the remaining agents ex-post guarantees with respect to their original APS value.

\end{appendix}

\end{document}